\documentclass[conference]{IEEEtran}
\IEEEoverridecommandlockouts
\usepackage[utf8]{inputenc}
\usepackage[T5]{fontenc}
\usepackage[vietnamese,english]{babel}
\addto\captionsenglish{}
\usepackage{cite}
\usepackage{algorithmic}
\usepackage{graphicx}
\usepackage{textcomp}
\usepackage{xcolor}
\usepackage{pgfplots}
\usepackage{float}
\def\BibTeX{{\rm B\kern-.05em{\sc i\kern-.025em b}\kern-.08em
    T\kern-.1667em\lower.7ex\hbox{E}\kern-.125emX}}
\usepackage{amsmath}
\usepackage{amsthm}   
\usepackage{amssymb}
\usepackage{amsfonts}
\usepackage{mathtools}
\usepackage{enumitem}
\usepackage{algorithm}
\usepackage[algo2e]{algorithm2e}
\usepackage{tikz}
\usetikzlibrary{arrows.meta}
\usepackage{bm}
\usepackage{csvsimple}
\usepackage{booktabs}
\usepackage{siunitx}
\theoremstyle{plain} 
\newtheorem{theorem}{Theorem}[section]     
\newtheorem{corollary}[theorem]{Corollary} 
\newtheorem{proposition}[theorem]{Proposition} 

\theoremstyle{definition} 
\newtheorem{definition}[theorem]{Definition}

\newtheorem{remark}[theorem]{Remark}
\newtheorem{problem}[theorem]{Problem}
\newtheorem{assumption}{Assumption}

\definecolor{mplblue}{RGB}{31, 119, 180}
\definecolor{mplorange}{RGB}{255, 127, 14}
\definecolor{mplgreen}{RGB}{44, 160, 44}
\definecolor{mplred}{RGB}{214, 39, 40}

\begin{document}

\title{A Novel Approach for the SDIR Epidemic Model on Online Social Networks}
\author{\IEEEauthorblockN{1\textsuperscript{st} Nguyen Hong Phuc}
	\IEEEauthorblockA{\textit{Faculty of Information Technology} \\
		\textit{Posts and Telecommunications}\\ \textit{Institute of Technology}\\
		Hanoi, Vietnam \\
		PhucNH.B25KH101@stu.ptit.edu.vn}
	\and
	\IEEEauthorblockN{2\textsuperscript{nd} Duong Khanh Ly}
	\IEEEauthorblockA{\textit{Faculty of Information Technology} \\
		\textit{Posts and Telecommunications}\\ \textit{Institute of Technology}\\
		Hanoi, Vietnam \\
		LyDK.B24CE169@stu.ptit.edu.vn}
     \and
     \IEEEauthorblockN{3\textsuperscript{rd} Hoang Phi Dung}
	\IEEEauthorblockA{\textit{Faculty of Fundamental Sciences} \\
		\textit{Posts and Telecommunications}\\ \textit{Institute of Technology}\\
		Hanoi, Vietnam \\
		dunghp@ptit.edu.vn}
}

\maketitle

\begin{abstract}
Information diffusion can be controlled by restricting or removing links (edges) in online social networks, as well as in real-world networks. To identify the most influential links to remove while minimizing diffusion, previous studies have proposed upper bounds for spreading processes in SIR and SIS models, using supermodularity and weighted matrices to identify critical links in contact networks. However, in some cases, existing upper bounds are not sufficiently tight to accurately capture the effect of important edges, as in the SDIR model of \cite{Khanh2026}. We therefore propose a tighter upper bound for controlling diffusion in the SDIR model by directly analyzing the dynamics of the two state vectors D and I in a $2N$-dimensional space. This approach yields an improved spectral-radius convergence condition and outperforms the previous method. Simulations on the synthetic Erd\H{o}s-R\'enyi network and the real-world Haslemere dataset using a Greedy edge-deletion algorithm demonstrate its effectiveness for influence minimization on social networks.
\end{abstract}


\section{Introduction}
Epidemic spreading models were introduced in the early twentieth century through the mathematical epidemiology work of W. O. Kermack and A. G. McKendrick \cite{Kermack1927}. Over the past two decades, seminal studies such as \cite{Albert2000,Pastor-Satorras2001} have shown that network topology plays a fundamental role in determining epidemic thresholds and dynamics (see, also \cite{Pastor-Satorras2015}). However, classical models often assume instantaneous or memoryless state transitions, which do not capture complex behavioral responses on modern online social platforms, such as user hesitation, delayed sharing, or deliberate decision-making, nor malware behavior in computer networks and the Internet \cite{Mieghem2009}. Consequently, epidemic models such as SIR, SIS, and their variants have become important tools in computer science, cybersecurity, malware analysis, online social networks, and complex-network theory \cite{Kempe2003,Ahn2013,Chen2014,Kempe2015,Shakarian2015,Nowzari2016,Silva2016,Pare2020,Ruhi2015,Yi2022,Khanh2026}. Fundamental problems in this area include influence maximization, influence minimization, influence blocking, and network anomaly detection.

Recently, interest in individual-based models has increased because complex networks effectively represent major technological platforms such as the Internet and online social networks \cite{Youssef2011,Sharkey2008,Yi2022,Pare2020}. Building on the work of Pare et al. \cite{Pare2020} and \cite{Yi2022}, Khanh et al. \cite{Khanh2026} proposed the SDIR model to describe user behavior on online social networks. On platforms such as Facebook, TikTok, and X (Twitter), users may hesitate and consider whether to share received information before making a decision. Thus, sharing depends on user behavior and is stochastic. In this paper, we study the SDIR model using a different mathematical approach to tighten the results in \cite{Khanh2026}.

\vspace{1em}
\noindent\textit{Related Works}

Individual-based models have recently been studied extensively, bringing the analysis of dynamical systems on complex networks closer to real technological platforms \cite{Youssef2011,Mieghem2009,Pastor-Satorras2015,Yi2022,Cho2024,Dung2026,Khanh2026}. Most of these models use either the Euler method \cite{Youssef2011,Pare2020} or mean-field approximation \cite{Mieghem2009,Pastor-Satorras2015,Yi2022,Khanh2026,Dung2026}. More recently, alternative approaches extending the threshold-based method of Wang \cite{Wang2003} have considered information or epidemic spreading within local communities of very large networks (see \cite{Dung2026-2}). Models incorporating a delay state during propagation have also received increasing attention \cite{Liu2022}. In contrast to influence maximization, influence minimization by reducing the number of infections has been considered in \cite{Pham2019,Yi2022,Chen2021,Xie2023,Khanh2026}.

\vspace{1em}
\noindent\textit{Contributions}

Unlike Khanh \textit{et al.}~\cite{Khanh2026}, where the two states
$\mathbf{x}(t)$ and $\mathbf{y}(t)$ are combined through the weighted quantity
$\mathbf{x}(t)+\mathbf{Q}\mathbf{y}(t)$, we directly analyze the two full state vectors in the
$2N$-dimensional space, thereby preserving their coupled dynamics.
This approach removes the dependence on the auxiliary weighting matrix
$\mathbf{Q}$ and yields a provably tighter spectral condition for model convergence in Theorem~\ref{sdir_previous_relation}.
We then adopt the edge-deletion approach in~\cite{Yi2022,Khanh2026,Dung2026} to minimize the number of infections, for which we derive the upper bound in Theorem~\ref{upper-bound} as a surrogate objective for the greedy edge-deletion algorithm.
Theorem~\ref{tighter-bound} further establishes that this bound is pointwise tighter than that proposed in~\cite{Khanh2026}.

\vspace{1em}
\noindent\textit{Outline}

The remainder of this paper is organized as follows. Section \ref{section-2} presents the SDIR model and the main assumptions. Sections \ref{section-3} and \ref{section-4} introduce the proposed approach and clarify its differences from the previous method through the convergence result (Theorem \ref{sdir_convergence}) and the spectral result (Theorem \ref{sdir_previous_relation}); an upper bound is then derived (Theorem \ref{upper-bound}) and the main problem is formulated (Problem \ref{problem}). Section \ref{section-5} presents the edge-deletion algorithm (Algorithm \ref{greedy}), and numerical results on a synthetic network and a real-world network are reported in Section \ref{section-6}. Finally, Section \ref{section-7} concludes the paper, and the proofs are provided in Appendix.

\vspace{1em}
\noindent\textit{Notation}

Let $\mathbb{R}$ denote the set of real numbers. For any positive integer $N$, we have $[N]=\{1,\ldots,N\}$, $\mathbf{I}$ denotes the identity matrix, $\operatorname{diag}(a_i)$ denotes the diagonal matrix with diagonal entries $a_i$, and $\rho(A)$ is the spectral radius of any matrix $A$. The Euclidean norm of a vector and the induced matrix $2$-norm are denoted by
$\|\cdot\|$, while $\|\cdot\|_1$ denotes the $1$-norm. For vectors or matrices of the same dimension, $\leq$ and $\geq$ are understood entrywise; for example, $\mathbf{X}\leq\mathbf{Y}$ if and only if $X_{ij}\leq Y_{ij}$ for all $i,j \in [N]$.

\section{Model Description}
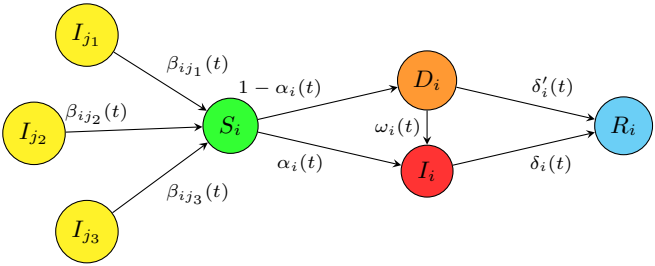
\label{section-2}

\begin{figure}[htbp]
    \centering
    \begin{tikzpicture}[
        >=stealth,
        node style/.style={circle, draw, minimum size=0.6cm, font=\small},
        S/.style={node style, fill=green!80},
        D/.style={node style, fill=orange!80},
        I/.style={node style, fill=red!80},
        R/.style={node style, fill=cyan!50},
        Ij/.style={node style, fill=yellow!90},
        lbl/.style={font=\scriptsize, inner sep=2pt}
    ]
        
        \node[S] at (-2.6, 0) (S) {$S_i$};
        \node[I] at (0, -0.6) (I) {$I_i$};
        \node[R] at (2.6, 0) (R) {$R_i$};
        \node[D] at (0, 0.6) (D) {$D_i$};
        
        \node[Ij] at (-4.5, 1.2) (Ij1) {$I_{j_1}$};
        \node[Ij] at (-5.2, -0.1) (Ij2) {$I_{j_2}$}; 
        \node[Ij] at (-4.5, -1.4) (Ij3) {$I_{j_3}$};

        \draw[->] (Ij1) -- (S) node[midway, above right, lbl] {$\beta_{ij_1}(t)$};
        \draw[->] (Ij2) -- (S) node[midway, above left, lbl] {$\beta_{ij_2}(t)$};
        \draw[->] (Ij3) -- (S) node[midway, below right, lbl] {$\beta_{ij_3}(t)$};

        \draw[->] (S) -- (I) node[midway, below left, lbl] {$\alpha_i(t)$};
        \draw[->] (S) -- (D) node[midway, above left, lbl] {$1-\alpha_i(t)$};

        \draw[->] (D) -- (I) node[midway, left, lbl] {$\omega_i(t)$};
        \draw[->] (D) -- (R) node[midway, above right, lbl] {$\delta'_i(t)$};
        
        \draw[->] (I) -- (R) node[midway, below right, lbl] {$\delta_i(t)$};
    \end{tikzpicture}
    \label{SDIR}
\end{figure}

The SDIR model proposed in \cite{Khanh2026} extends the SIR model to describe information diffusion on social networks. The $D$ (Delayable) state represents users (nodes) who have received information but hesitate to spread it immediately and may subsequently transition to state $I$ or $R$.

Consider a spreading process on a digraph $G = (V, E)$ with $N$ nodes. At each time step $t$, the state of node $i\in V$ is represented by the indicator variables $S_i(t), D_i(t), I_i(t)$ and $R_i(t)\in\{0,1\}$, corresponding to the Susceptible ($S$), Delayable ($D$), Infected ($I$), and Recovered ($R$) states, respectively. Since each node occupies exactly one state at a time, $S_i(t) + D_i(t) + I_i(t) + R_i(t) = 1, \forall i\in V, \forall t>0$. A neighboring node $j$ of $i$ in state $I$ can infect node $i$ with probability $\beta_{ij}(t)$; upon successful infection, node $i$ moves directly to state $I$ with probability $\alpha_i(t)$ or to the delay state $D$ with probability $1-\alpha_i(t)$. For node $i$ in state $D$, draw $p\sim U[0,1]$. Node $i$ moves to $I$ if $p<\omega_i(t)$, to $R$ if $\omega_i(t)\leq p<\omega_i(t)+\delta'_i(t)$, and otherwise remains in $D$. Finally, node $i$ recovers to $R$ at rate $\delta_i(t)$ from state $I$. We assume that the parameters $\beta_{ij}(t), \alpha_i(t), \omega_i(t), \delta_i(t)$ and $\delta'_i(t)$ are independent random variables with time-invariant distributions.

Based on the information-spreading model in \cite{Khanh2026}, taking expectations on both sides and applying mean-field approximation to the infection probability yields the deterministic SDIR model.
\begin{align}\label{eq:SDIR}
	\mathbf{x}(t+1) &= (\mathbf{I}-\mathbf{D}+\mathbf{A}\mathbf{S}(t)\mathbf{B})\mathbf{x}(t) \\ 
	& + \mathbf{W}\mathbf{y}(t) \notag\\
	\mathbf{y}(t+1) &= (\mathbf{I}-\mathbf{A})\mathbf{S}(t)\mathbf{B}\mathbf{x}(t) \notag\\ 
	& + (\mathbf{I}-\mathbf{W}-\mathbf{D}')\mathbf{y}(t) \\ \nonumber 
	\mathbf{r}(t+1) &= \mathbf{D}\mathbf{x}(t)+\mathbf{D}'\mathbf{y}(t)+\mathbf{r}(t)  \nonumber 
\end{align}
Here, the vectors $\mathbf{x}(t), \mathbf{y}(t)$, and $\mathbf{r}(t)$ contain the $i$th entries $\mathbb{E}[I_i(t)], \mathbb{E}[D_i(t)]$, and $\mathbb{E}[R_i(t)]$, respectively, for $\forall i \in [N]$. The matrix $\mathbf{S}(t)=\operatorname{diag}(s_i(t))$, where $s_i(t)=\mathbb{E}[S_i(t)]$ $\forall i \in [N]$, and the matrix $\mathbf{B}$ has entries $B_{ij}=\mathbb{E}[\beta_{ij}(t)]$. The remaining diagonal matrices are $\mathbf{D}=\operatorname{diag}(\mathbb{E}[\delta_i(t)]), \mathbf{A}=\operatorname{diag}(\mathbb{E}[\alpha_i(t)]), \mathbf{W}=\operatorname{diag}(\mathbb{E}[\omega_i(t)])$, and $\mathbf{D}'=\operatorname{diag}(\mathbb{E}[\delta'_i(t)])$. We also assume that $\sum_{j=1}^N B_{ij}<1, \forall i\in [N]$.

\begin{remark}
Unlike~\cite{Khanh2026}, where $\mathbf{S}(t)$ is replaced by $\mathbf{S}(0)$ to obtain a linear recurrence, we retain the time dependence of $\mathbf{S}(t)$. This directly reflects the variation of the susceptible state during propagation.
\end{remark}

\begin{assumption}\label{assp1}
$D_i\leq D'_i$ for all $i\in[N]$.
\end{assumption}
This assumption is equivalent to Assumption~1 in~\cite{Khanh2026}.
\begin{assumption}\label{assp3}
$0<W_i+D'_i\leq1$ for all $i\in[N]$.
\end{assumption}
The quantity $W_i+D'_i$ represents the probability that a node leaves the delay state $D$ at a given time step, either by transitioning to state $I$ or to state $R$. Thus, this assumption ensures that the transition probabilities are valid and that a node eventually leaves state $D$ after the delay process.

\section{Global Convergence of Deterministic SDIR Model}\label{section-3}

In the D-SIR model of Yi \textit{et al.}~\cite{Yi2022},
$\mathbf{x}(t+1)$ is determined directly by a linear system based on $\mathbf{x}(t)$. In the
SDIR model, the delay state $D$ couples
$\mathbf{x}(t)$ and $\mathbf{y}(t)$. We therefore analyze these two state vectors jointly.

For convenience, let $\mathbf{F}=\mathbf{S}(0)\mathbf{B}$.

For $\mathbf{x}(t),\mathbf{y}(t)\in\mathbb{R}^N$, define
$
\mathbf{v}(t)
\overset{\mathrm{def}}{=}
\begin{bmatrix}
\mathbf{x}(t)\\
\mathbf{y}(t)
\end{bmatrix}
\in\mathbb{R}^{2N}
$ and consider the matrix $\mathbf{M}
=
\begin{bmatrix}
\mathbf{I}-\mathbf{D}+\mathbf{A}\mathbf{F}
&
\mathbf{W}\\
(\mathbf{I}-\mathbf{A})\mathbf{F}
&
\mathbf{I}-\mathbf{W}-\mathbf{D}'
\end{bmatrix}$ as the state-transition matrix from $\mathbf{v}(0)$ to $\mathbf{v}(1)$.

\begin{theorem}\label{sdir_convergence}
Consider the SDIR model satisfying $\rho(\mathbf{M})<1$. Then
$\widehat{\mathbf{v}}=\mathbf{0}$ is the unique equilibrium of $(\mathbf{x},\mathbf{y})$; specifically, $\lim_{t\to \infty}\mathbf{x}(t)=\lim_{t\to \infty}\mathbf{y}(t)=\mathbf{0}$, and the system $(\mathbf{x}(t),\mathbf{y}(t))$ is exponentially stable at the origin for every initial state satisfying the model conditions.
\end{theorem}

We next compare the convergence condition based on the $2N\times2N$ matrix
$\mathbf{M}$ with that in~\cite{Khanh2026}.
Following~\cite{Khanh2026}, under Assumption~\ref{assp1}, we fix throughout this paper $\mathbf{Q}=\operatorname{diag}(q_i)$ with
$q_i\in\left[\frac{W_i}{W_i+D'_i-D_i},1\right]$, provided that $W_i+D'_i-D_i>0$. If $W_i+D'_i-D_i=0$, which gives $W_i=0$, we choose any $q_i\in(0,1]$. Thus, $q_i>0$ for all $i\in[N]$. Define
$\mathbf{G}_Q=\mathbf{A}+\mathbf{Q}(\mathbf{I}-\mathbf{A})$ and
$\mathbf{M}_Q=\mathbf{I}-\mathbf{D}+\mathbf{G}_Q\mathbf{F}$ is the comparison matrix used for evaluating the convergence condition in~\cite{Khanh2026}.

\begin{theorem}\label{sdir_previous_relation}
Under Assumption~\ref{assp1}, we have $\rho(\mathbf{M})\leq\rho(\mathbf{M}_Q)$.
\end{theorem}

Moreover, it was shown in~\cite{Khanh2026} that, under
Assumption~\ref{assp1}, $
\rho(\mathbf{M}_Q)
\leq
\rho(\mathbf{M}_{\mathrm{SIR}})$ where
$\mathbf{M}_{\mathrm{SIR}}
=\mathbf{I}-\mathbf{D}+\mathbf{F}$, as given in~\cite{Yi2022}. Therefore, we have the following result.

\begin{corollary}\label{sdir_sir_relation}
Under Assumption~\ref{assp1}, we have $
\rho(\mathbf{M})
\leq
\rho(\mathbf{M}_Q)
\leq
\rho(\mathbf{M}_{\mathrm{SIR}})$. Consequently, if $\rho(\mathbf{M}_{\mathrm{SIR}})<1$, then
$\rho(\mathbf{M})<1$.
\end{corollary}
\section{Problem and Bounding Function}\label{section-4}
\subsection{Main Problem}

To quantify the infection level of the network, for each
$i\in[N]$, define $m_i(t)=x_i(t)+y_i(t)+r_i(t)$ and let $\mathbf{m}(t)=\mathbf{x}(t)+\mathbf{y}(t)+\mathbf{r}(t)$. Assume the initial state satisfies $m_i(0)=x_i(0)+y_i(0)+r_i(0)\in[0,1]$ for all $i\in[N]$. From \eqref{eq:SDIR}, we have
$
m_i(t+1)-m_i(t)
=
(1-m_i(t))
\sum_{j=1}^{N}B_{ij}x_j(t).
$ Since $0\leq x_j(t)\leq m_j(t)\leq1$, and
$\sum_{j=1}^{N}B_{ij}<1$, we have
$
0\leq
\sum_{j=1}^{N}B_{ij}x_j(t)
<1.
$ If $m_i(t)\in[0,1]$, then $0\leq1-m_i(t)\leq1$, and hence
$
0\leq
(1-m_i(t))
\sum_{j=1}^{N}B_{ij}x_j(t)
\leq
1-m_i(t).
$
Hence, $
0\leq m_i(t)\leq m_i(t+1)\leq1.
$ By induction, $m_i(t)\in[0,1]$ for all $t\geq0$, and $m_i(t)$ is
monotone non-decreasing over time. Thus, $m_i(t)$ can be interpreted as
the probability that node $i$ has left state $S$ by time $t$,
i.e., the probability that node $i$ is in one of the states $D$, $I$, or $R$. Moreover, since $s_i(t)=1-m_i(t)$ $\forall t$, $\mathbf{S}(t)$ is entrywise monotone non-increasing over time, and hence $\mathbf{S}(t)\leq \mathbf{S}(0)$. Let $\mathbf{m}^*\in\mathbb{R}^N$ denote the vector with entries $m_i^*=\sup_{t\geq0}m_i(t)$ $\forall i\in[N]$.

\begin{definition}
For a deleted edge set $P$, define
\begin{align}
\Phi(P)
\overset{\mathrm{def}}{=}
\|\mathbf{m}^*-\mathbf{m}(0)\|_1
\label{cumulative}
\end{align}
as the cumulative number of new infections after the propagation process terminates.
\end{definition}

\begin{problem}\label{problem}
Given a digraph $G=(V,E)$ and initial states
$\mathbf{x}(0),\mathbf{y}(0),\mathbf{r}(0)$ such that
$\mathbf{x}(0)+\mathbf{y}(0)+\mathbf{r}(0)\in[0,1]^N$, let
$\mathcal{Q}\subseteq E$ be a candidate set of edges and let
$k\leq|\mathcal{Q}|$ be a positive integer. Find a set $P^*\subseteq\mathcal{Q}$ with
$|P^*|\leq k$ such that
\[
P^*
\in
\operatorname*{argmin}_{P\subseteq\mathcal{Q},\,|P|\leq k}
\Phi(P).
\]
\end{problem}
\begin{remark} Finding an optimal solution $P^*$ to Problem \ref{problem} is NP-hard. The proof is analogous to \cite[Theorem 4.4]{Yi2022}.
\end{remark}

Since Problem~\ref{problem} is NP-hard, directly optimizing $\Phi(P)$ is impractical for large networks. Therefore, when the model convergence condition holds, i.e., the probability of each node being infected converges to zero, we construct $\widehat{\Phi}(P)$ as an upper bound on $\Phi(P)$ that is monotone and supermodular with respect to the deleted edge set $P$. This upper bound serves as a surrogate objective, enabling an efficient greedy edge-selection algorithm.

\subsection{Supermodular Upper Bound}

For a deleted edge set $P\subseteq \mathcal{Q}$, let
$\mathbf{B}_{-P}$ denote the matrix obtained from $\mathbf{B}$ by
setting $B_{ij}=0$ for every edge $(j,i)\in P$, and define $
\mathbf{F}_{-P}=\mathbf{S}(0)\mathbf{B}_{-P}$. Correspondingly, define the comparison matrix $
\mathbf{M}_{-P}
=
\begin{bmatrix}
\mathbf{I}-\mathbf{D}+\mathbf{A}\mathbf{F}_{-P}
&
\mathbf{W}\\
(\mathbf{I}-\mathbf{A})\mathbf{F}_{-P}
&
\mathbf{I}-\mathbf{W}-\mathbf{D}'
\end{bmatrix}$.

Under Assumption~\ref{assp3}, the matrix $\mathbf{W}+\mathbf{D}'$ is invertible. Define $\mathbf{T}=\mathbf{W}(\mathbf{W}+\mathbf{D}')^{-1}$ and $\mathbf{\Theta}=\mathbf{A}+\mathbf{T}(\mathbf{I}-\mathbf{A})$.

\begin{theorem}\label{upper-bound}
If $\rho(\mathbf{M}_{-P})<1$, then the cumulative number of new infections
after deleting the edge set $P$ satisfies
\begin{align}\label{eq:upper_bound}
\Phi(P) \leq \widehat{\Phi}(P)=1^\top
\mathbf{F}_{-P}
\left(
\mathbf{D}
-
\mathbf{\Theta}\mathbf{F}_{-P}
\right)^{-1}
\left(
\mathbf{x}(0)
+
\mathbf{T}\mathbf{y}(0)
\right).
\end{align}
\end{theorem}

\begin{remark}
If the original model satisfies $\rho(\mathbf{M})<1$, then for every
$P\subseteq \mathcal{Q}$, we have
$\mathbf{B}_{-P}\leq\mathbf{B}$, and therefore
$\mathbf{M}_{-P}\leq\mathbf{M}$.
Since $\mathbf{M}_{-P}$ and $\mathbf{M}$ are nonnegative matrices, the monotonicity of the spectral radius~\cite{Horn2012} gives
$\rho(\mathbf{M}_{-P})\leq\rho(\mathbf{M})<1$. Hence, the condition in Theorem~\ref{upper-bound} holds for every deleted edge set $P\subseteq \mathcal{Q}$.
\end{remark}

We next compare the proposed upper bound with that in~\cite{Khanh2026}.
For $P\subseteq\mathcal{Q}$, let
$\mathbf{M}_{Q,-P}
=\mathbf{I}-\mathbf{D}+\mathbf{G}_{Q}\mathbf{F}_{-P}$.  Under the condition $\rho(\mathbf{M}_{Q,-P})<1$, the upper bound in~\cite{Khanh2026} is given by
\begin{align}
\Phi_Q(P)
=
1^\top
\mathbf{F}_{-P}
\left(
\mathbf{D}
-
\mathbf{G}_{Q}\mathbf{F}_{-P}
\right)^{-1}
\left(
\mathbf{x}(0)
+
\mathbf{Q}\mathbf{y}(0)
\right).
\label{eq:old_upper_bound}
\end{align}

\begin{remark}
$\Phi_Q$ is monotone non-increasing and supermodular with respect to
the deleted edge set $P\subseteq \mathcal{Q}$, as given in~\cite{Khanh2026}.
\label{remarkofsuper}
\end{remark}

\begin{theorem}\label{tighter-bound}
Under Assumption~\ref{assp1}, if $\rho(\mathbf{M}_{Q,-P}) < 1$ then
\begin{align*}
\Phi(P)\leq\widehat{\Phi}(P)\leq\Phi_Q(P),
\end{align*}
where $\widehat{\Phi}(P)$ and $\Phi_Q(P)$ are defined in \eqref{eq:upper_bound} and \eqref{eq:old_upper_bound}, respectively.
\end{theorem}
Although the deterministic SDIR model in this paper retains $\mathbf{S}(t)$ in the recurrence, the upper bound in~\cite{Khanh2026} remains a valid upper bound for the cumulative number of new infections considered here. Moreover, the theorem shows that jointly considering the two state vectors $\mathbf{x}(t)$ and $\mathbf{y}(t)$ yields an upper bound tighter than the upper bound proposed in~\cite{Khanh2026}.

\begin{proposition}\label{supermodular}
If $\rho(\mathbf{M})<1$, then
$\widehat{\Phi}(P)$ is monotone non-increasing and supermodular with respect to
the deleted edge set $P\subseteq \mathcal{Q}$.
\end{proposition}
All proofs are provided in the Appendix.
\section{Edge Deletion Algorithm}\label{section-5}
Consider the objective function
$f\in\{\widehat{\Phi},\Phi_Q\}$, where
$\widehat{\Phi}$ is the upper bound proposed in this paper and
$\Phi_Q$ is the upper bound in~\cite{Khanh2026}.
Since both functions are monotone non-increasing and supermodular with respect to the
deleted edge set, we apply the following Greedy algorithm to select $k$ edges from
the candidate edge set $\mathcal{Q}$.

\begin{algorithm}[H]
	\caption{Greedy Algorithm (GA)}
	\label{greedy}
	\begin{flushleft}
	\textbf{Input:} A function
	$f\in\{\widehat{\Phi},\Phi_Q\}$, a graph $G$, initial states,
	a candidate edge set $\mathcal{Q}$, and an integer $k$.
	
	\textbf{Output:} An edge set $P\subseteq\mathcal{Q}$ of size $k$.
	
	Initialize $P\leftarrow \emptyset$
	
	\For{$i = 1$ \KwTo $k$}{
		Compute $f(P\cup \{e\})$ for each
		$e\in\mathcal{Q}\backslash P$\\
		$e^\star \leftarrow
		\operatorname*{argmax}_{e\in\mathcal{Q}\backslash P}
		\big(f(P)-f(P\cup \{e\})\big)$\\
		$P\leftarrow P\cup \{e^\star\}$
	}
	
	\Return $P$
	\end{flushleft}
\end{algorithm}
The Greedy Algorithm (GA) that uses $f=\widehat{\Phi}$ is called the
Improved Greedy Algorithm (\textbf{IGA}), while the version using $f=\Phi_Q$ is called the Baseline Greedy Algorithm (\textbf{BGA}). Based on Proposition \ref{supermodular} and Remark~\ref{remarkofsuper}, the solution returned by the greedy algorithm guarantees an approximation ratio of $(1-\frac{1}{e})$ for the function $
f(\emptyset)-f(\cdot)$, with respect to the optimal solution.
\section{Numerical Experiments}\label{section-6}
In this section, we perform numerical simulations of \eqref{eq:SDIR} to validate the theoretical results. We compare the proposed Improved Greedy Algorithm (\textbf{IGA}) with the following three algorithms
\begin{itemize}
	\item \textbf{Random} \cite{Callaway2000}: Randomly delete one edge at each iteration.
	\item \textbf{Max-Degree} \cite{Albert2000}: Delete an edge incident to the highest-degree node at each iteration.
	\item \textbf{BGA}: Algorithm~\ref{greedy} applied to $f=\Phi_Q$ defined in \eqref{eq:old_upper_bound}.
\end{itemize}

We first evaluate the proposed method on an Erd\H{o}s-R\'enyi (ER) network with $N=700$ and $p=0.023$. The initial states are nonzero at $|S|=7$ seed nodes, with $x_i(0)\in[0.80,0.85]$ and $y_i(0)\in[0,0.05]$. We sample $B_{ij}\in[0.024,0.034]$, $W_i\in[0.15,0.35]$, and $D_i\in[0.30,0.48]$, with $W_i+D'_i\leq0.95$. For $A_i$, the seed nodes, $45$ randomly selected non-seed nodes, and the remaining nodes use $[0.50,0.78]$, $[0.65,0.87]$, and $[0.15,0.35]$, respectively. We set $|\mathcal{Q}|=3500$ and $k=750$; results are shown in Fig.~\ref{fig:er_network}.

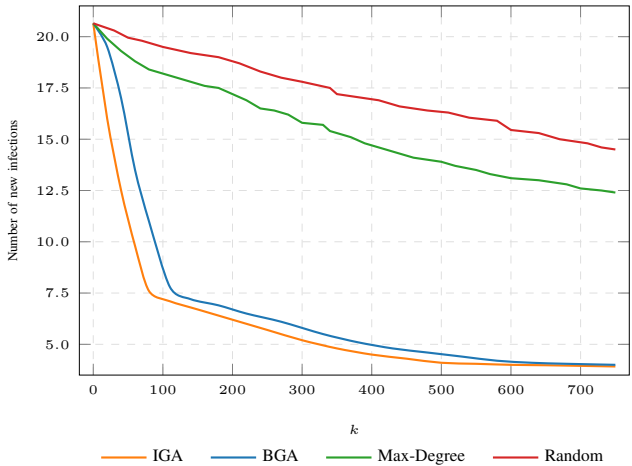
\begin{figure}[H]
\centering
\begin{tikzpicture}
\begin{axis}[
    width=\linewidth,
    height=0.73\linewidth,
    xlabel={$k$},
    ylabel={Number of new infections},
    ylabel near ticks,
    xmin=-20, xmax=770,
    ymin=3.5, ymax=21.5,
    xtick={0, 100, 200, 300, 400, 500, 600, 700},
    ytick={5.0, 7.5, 10.0, 12.5, 15.0, 17.5, 20.0},
    yticklabel style={
        /pgf/number format/fixed,
        /pgf/number format/precision=1,
        /pgf/number format/zerofill
    },
    grid=both,
    grid style={dashed, gray!25, line width=0.1pt},
    axis line style={line width=0.1pt},
    tick style={line width=0.1pt},
    label style={font=\tiny},
    tick label style={font=\tiny},
    title style={font=\tiny\bfseries, yshift=-0.5ex},
    line width=0.8pt,
    legend style={
        at={(0.5,-0.17)},     
        anchor=north,
        legend columns=-1,    
        font=\fontsize{6.5}{7.5}\selectfont,
        draw=none,             
        fill=none,             
        inner sep=2pt,
        /tikz/every even column/.append style={column sep=8pt} 
    }
]

\addplot[color=mplorange, smooth] coordinates {
    (0, 20.65) (20, 16.0) (40, 12.5) (60, 9.8) (80, 7.6) 
    (110, 7.1) (150, 6.7) (200, 6.2) (250, 5.7) (300, 5.2) 
    (350, 4.8) (400, 4.5) (450, 4.3) (500, 4.1) (550, 4.05) 
    (600, 4.0) (650, 3.98) (700, 3.95) (750, 3.92)
};
\addlegendentry{IGA}  

\addplot[color=mplblue, smooth] coordinates {
    (0, 20.65) (20, 19.5) (40, 17.0) (60, 13.5) (80, 11.0) 
    (110, 7.8) (140, 7.2) (180, 6.9) (220, 6.5) (270, 6.1) 
    (330, 5.5) (380, 5.1) (430, 4.8) (480, 4.6) (530, 4.4) 
    (580, 4.2) (630, 4.1) (680, 4.05) (750, 4.0)
};
\addlegendentry{BGA}  

\addplot[color=mplgreen] coordinates {
    (0, 20.65) (20, 19.9) (40, 19.3) (60, 18.8) (80, 18.4) (100, 18.2) 
    (120, 18.0) (140, 17.8) (160, 17.6) (180, 17.5) (200, 17.2) (220, 16.9) 
    (240, 16.5) (260, 16.4) (280, 16.2) (300, 15.8) (330, 15.7) (340, 15.4) 
    (370, 15.1) (390, 14.8) (410, 14.6) (430, 14.4) (460, 14.1) (500, 13.9) 
    (520, 13.7) (550, 13.5) (570, 13.3) (600, 13.1) (640, 13.0) (680, 12.8) 
    (700, 12.6) (730, 12.5) (750, 12.4)
};
\addlegendentry{Max-Degree}  

\addplot[color=mplred] coordinates {
    (0, 20.65) (30, 20.3) (50, 19.95) (70, 19.8) (100, 19.5) (140, 19.2) 
    (180, 19.0) (210, 18.7) (240, 18.3) (270, 18.0) (300, 17.8) (340, 17.5) 
    (350, 17.2) (390, 17.0) (410, 16.9) (440, 16.6) (480, 16.4) (510, 16.3) 
    (540, 16.05) (580, 15.9) (600, 15.45) (640, 15.3) (670, 15.0) (710, 14.8) 
    (730, 14.6) (750, 14.5)
};
\addlegendentry{Random}  

\end{axis}
\end{tikzpicture}
\caption{Performance evaluation of edge deletion on Erd\H{o}s--R\'enyi Network.}
\label{fig:er_network}
\end{figure}

To further validate the effectiveness of our method on real-world data, we next consider the Haslemere dataset collected in the UK through the BBC Pandemic project~\cite{Firth2020,Klepac2018}.  Following the parameters in \cite{Khanh2026}, the infection parameter is adjusted such that $B_{ij}\in[0.056, 0.063]$. We compare the algorithms in terms of performance using the same deletion budget $k=210$ from a candidate set $\mathcal{Q}$ of $520$ edges. The corresponding results are shown in Fig.~\ref{fig:haslemere_network}.

\begin{figure}[H]
\centering
\begin{tikzpicture}
\begin{axis}[
    width=\linewidth,
    height=0.75\linewidth,
    xlabel={$k$},
    ylabel={Number of new infections},
    ylabel near ticks,
    xmin=-8, xmax=220,
    ymin=1.7, ymax=9.1,
    xtick={0, 50, 100, 150, 200},
    ytick={2, 3, 4, 5, 6, 7, 8, 9},
    grid=both,
    grid style={dashed, gray!25, line width=0.1pt},
    axis line style={line width=0.1pt},             
    tick style={line width=0.1pt},                  
    legend style={
        at={(0.5,-0.18)},     
        anchor=north,
        legend columns=-1,      
        font=\fontsize{6.5}{7.5}\selectfont,
        draw=none,             
        fill=none,             
        inner sep=2pt,
        /tikz/every even column/.append style={column sep=12pt} 
    },
    label style={font=\tiny},
    tick label style={font=\tiny},
    title style={font=\tiny\bfseries, yshift=-0.5ex},
    line width=0.8pt                             
]
\addplot[color=mplorange] coordinates {
    (0, 8.8) (10, 5.1) (20, 3.45) (30, 2.88) (40, 2.61) 
    (50, 2.44) (60, 2.35) (70, 2.29) (80, 2.25) (90, 2.21) 
    (100, 2.19) (120, 2.14) (150, 2.09) (180, 2.05) (210, 2.04)
};
\addlegendentry{IGA}
\addplot[color=mplblue] coordinates {
    (0, 8.8) (10, 6.2) (20, 4.35) (30, 3.3) (40, 2.9) 
    (50, 2.67) (60, 2.52) (70, 2.39) (80, 2.31) (90, 2.25) 
    (100, 2.22) (120, 2.18) (150, 2.11) (180, 2.08) (210, 2.05)
};
\addlegendentry{BGA}
\addplot[color=mplgreen] coordinates {
    (0, 8.75) (10, 7.65) (30, 6.57) (50, 6.25) (70, 6.10) 
    (80, 6.00) (90, 5.58) (110, 5.45) (120, 5.32) (150, 5.14) 
    (160, 4.80) (180, 4.72) (190, 4.70) (205, 4.60) (210, 4.43)
};
\addlegendentry{Max-Degree}
\addplot[color=mplred] coordinates {
    (0, 8.8) (10, 8.7) (20, 8.56) (70, 8.22) (80, 8.03) 
    (90, 7.98) (110, 7.74) (120, 7.15) (130, 6.54) (140, 6.28) 
    (150, 5.86) (160, 5.78) (170, 5.57) (180, 5.50) (190, 4.76) (210, 4.66)
};
\addlegendentry{Random}

\end{axis}
\end{tikzpicture}
\caption{Performance evaluation of edge deletion on Haslemere Network.}
\label{fig:haslemere_network}
\end{figure}
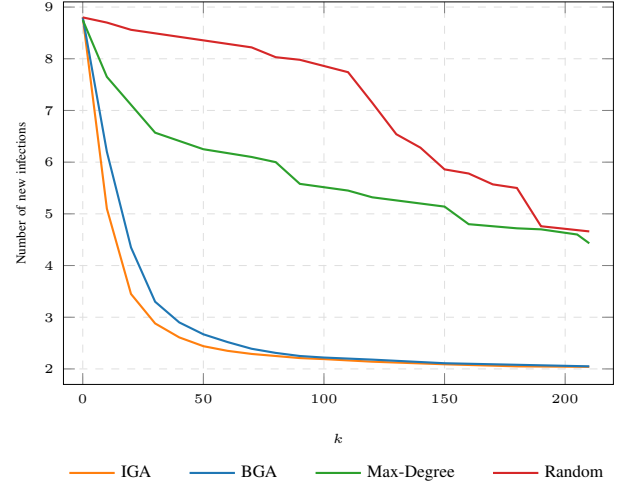

The edge-deletion process of Algorithm~\ref{greedy} in Fig.~\ref{fig:er_network} and Fig.~\ref{fig:haslemere_network}, using the two upper bounds $\widehat{\Phi}(\cdot)$ and $\Phi_Q(\cdot)$, guides edge selection substantially more effectively than \textbf{Random} and \textbf{Max-Degree}. In particular, the improved algorithm, \textbf{IGA}, using $\widehat{\Phi}(\cdot)$ reduces infections faster than \textbf{BGA}, which uses $\Phi_Q(\cdot)$ from~\cite{Khanh2026}. This observation is consistent with Theorem~\ref{tighter-bound}, which shows that $\widehat{\Phi}(\cdot)$ provides a tighter upper bound than $\Phi_Q(\cdot)$ and hence a more informative surrogate for ranking candidate edges by marginal reduction.
Although the proposed analysis is performed in the $2N$-dimensional space, the resulting upper bound in Theorem~\ref{upper-bound} only involves the inverse of an $N\times N$ matrix, as does $\Phi_Q$ in~\cite{Khanh2026}. Moreover, using the rank-one update strategy in~\cite{Yi2022}, both IGA and BGA can be implemented with computational complexity
$O\!\left(N^3+k(N^2+|\mathcal{Q}|N)\right)$. 

Finally, we compare two quantities that upper-bound the infection level of the two comparison systems on the same Haslemere dataset and with the same infection parameters as in Fig. \ref{fig:haslemere_network}. Specifically, we consider $\|\mathbf{L}_Q\mathbf{M}^t\mathbf{v}(0)\|_1$
for the proposed method and
$\|\mathbf{M}_Q^t(\mathbf{x}(0)+\mathbf{Q}\mathbf{y}(0))\|_1$
for the method in~\cite{Khanh2026}, where $\mathbf{L}_Q=[\mathbf{I}\ \mathbf{Q}]$ is used to ensure that the two quantities have the same initial value
$\|\mathbf{x}(0)+\mathbf{Q}\mathbf{y}(0)\|_1$. Their decay rates are then governed by the spectral radii of the corresponding matrices, namely $\rho(\mathbf{M})$ and
$\rho(\mathbf{M}_Q)$. 

In Fig.~\ref{fig:convergence}, the quantity corresponding to the proposed method decreases and converges to zero faster than that of the method in~\cite{Khanh2026}. This result is consistent with Theorem~\ref{sdir_previous_relation}, which gives
$\rho(\mathbf{M})\leq\rho(\mathbf{M}_Q)$.

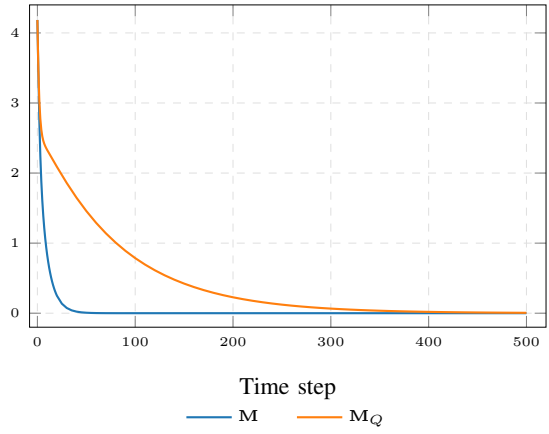
\begin{figure}[H]
\centering
\begin{tikzpicture}
\begin{axis}[
    width=0.95\linewidth,
    height=0.66\linewidth,
    xlabel={Time step},
    xmin=-8, xmax=520,
    ymin=-0.2, ymax=4.4,
    xtick={0, 100, 200, 300, 400, 500},
    ytick={0, 1, 2, 3, 4},
    grid=both,
    grid style={dashed, gray!25, line width=0.1pt},
    axis line style={line width=0.1pt},
    tick style={line width=0.1pt},
    legend style={
        at={(0.5,-0.22)},
        anchor=north,
        legend columns=2,
        font=\fontsize{6.5}{7.5}\selectfont,
        draw=none,
        fill=none,
        inner sep=2pt,
        /tikz/every even column/.append style={column sep=12pt}
    },
    label style={font=\small},
    tick label style={font=\tiny},
    title style={font=\small\bfseries, yshift=-0.5ex},
    line width=0.8pt
]

\addplot[color=mplblue] coordinates {
    (0,4.1849579285)
    (1,3.4522364305)
    (2,2.7466367554)
    (3,2.2467998660)
    (4,1.8842842067)
    (5,1.6099433038)
    (6,1.3935014403)
    (7,1.2167233943)
    (8,1.0684821867)
    (9,0.9417895868)
    (10,0.8320838349)
    (11,0.7362484077)
    (12,0.6520465126)
    (13,0.5777923684)
    (14,0.5121582921)
    (15,0.4540600879)
    (16,0.4025877096)
    (17,0.3569621571)
    (18,0.3165076504)
    (19,0.2806327931)
    (20,0.2488171246)
    (25,0.1362579924)
    (30,0.0745589287)
    (35,0.0407739649)
    (40,0.0222886476)
    (45,0.0121801485)
    (50,0.0066546510)
    (55,0.0036351916)
    (60,0.0019855322)
    (65,0.0010843947)
    (70,0.0005921998)
    (75,0.0003233901)
    (80,0.0001765908)
    (85,0.0000964265)
    (90,0.0000526519)
    (95,0.0000287491)
    (100,0.0000156974)
    (103,0.0000109181)
    (500,0.0000000000)
};
\addlegendentry{$\mathbf{M}$}

\addplot[color=mplorange] coordinates {
    (0,4.1849579285)
    (1,3.4616381388)
    (2,3.0294703030)
    (3,2.7732260111)
    (4,2.6197822713)
    (5,2.5249832640)
    (6,2.4630027162)
    (7,2.4190710841)
    (8,2.3848880095)
    (9,2.3558595310)
    (10,2.3294824146)
    (20,2.0880583965)
    (30,1.8583810473)
    (40,1.6481285612)
    (50,1.4589714855)
    (60,1.2901944382)
    (70,1.1402839609)
    (80,1.0074708072)
    (90,0.8899725259)
    (100,0.7861044168)
    (110,0.6943242102)
    (120,0.6132435294)
    (130,0.5416236686)
    (140,0.4783647203)
    (150,0.4224924966)
    (160,0.3731453162)
    (170,0.3295615509)
    (180,0.2910682574)
    (190,0.2570709666)
    (200,0.2270445863)
    (210,0.2005253292)
    (220,0.1771035667)
    (230,0.1564175099)
    (240,0.1381476251)
    (250,0.1220116995)
    (260,0.1077604827)
    (270,0.0951738372)
    (280,0.0840573376)
    (290,0.0742392679)
    (300,0.0655679690)
    (310,0.0579094955)
    (320,0.0511455475)
    (330,0.0451716425)
    (340,0.0398955019)
    (350,0.0352356253)
    (360,0.0311200319)
    (370,0.0274851482)
    (380,0.0242748265)
    (390,0.0214394769)
    (400,0.0189353019)
    (410,0.0167236198)
    (420,0.0147702666)
    (430,0.0130450692)
    (440,0.0115213783)
    (450,0.0101756576)
    (460,0.0089871192)
    (470,0.0079374059)
    (480,0.0070103005)
    (490,0.0061914830)
    (500,0.0054683051)
};
\addlegendentry{$\mathbf{M}_Q$}

\end{axis}
\end{tikzpicture}
\caption{Convergence comparison between $\mathbf{M}$ and $\mathbf{M}_Q$ formulations.}
\label{fig:convergence}
\end{figure}

\section{Conclusion}\label{section-7}
Rather than relying on a weighted combination of the two states as in previous studies, this paper directly analyzes the SDIR model in the $2N$-dimensional space. This formulation yields a tighter spectral convergence condition and a tighter upper bound for influence minimization, which is incorporated into a greedy edge-deletion algorithm. Experiments on the synthetic Erd\H{o}s-R\'enyi network and the real-world Haslemere dataset demonstrate the effectiveness of the resulting Improved Greedy Algorithm (\textbf{IGA}) in reducing information spread. Future work will investigate broader network topologies, parameter sensitivity, and scalable implementations for larger networks.

\section*{Appendix}\label{Appendix}

\subsection{Proof of Theorem~\ref{sdir_convergence}}

\begin{proof}
Since $\mathbf{S}(t)\leq\mathbf{S}(0)$ for all $t\geq 0$, from \eqref{eq:SDIR} we have $
\mathbf{v}(t+1)
=
\begin{bmatrix}
\mathbf{I}-\mathbf{D}+\mathbf{A}\mathbf{S}(t)\mathbf{B}
&
\mathbf{W}\\
(\mathbf{I}-\mathbf{A})\mathbf{S}(t)\mathbf{B}
&
\mathbf{I}-\mathbf{W}-\mathbf{D}'
\end{bmatrix}
\mathbf{v}(t)
\leq
\mathbf{M}\mathbf{v}(t)$. Since $\mathbf{M}\geq0$, induction yields
$\mathbf{v}(t)\leq\mathbf{M}^t\mathbf{v}(0)$ for all $t\geq0$.

Since $\rho(\mathbf{M})<1$, there exist $\gamma$ and $C>0$ such that
$\rho(\mathbf{M})<\gamma<1$ and
$\|\mathbf{M}^t\|\leq C\gamma^t$.
Thus, $
\|\mathbf{v}(t)\|
\leq
C\gamma^t\|\mathbf{v}(0)\|$. Since $\lim_{t\rightarrow\infty}\gamma^t=0$,
$\mathbf{v}(t)$ converges exponentially to $\mathbf{0}$. Specifically,
$\lim_{t\to\infty}\mathbf{x}(t)=\mathbf{0}$ and
$\lim_{t\to\infty}\mathbf{y}(t)=\mathbf{0}$.

Hence, $\widehat{\mathbf{v}}=\mathbf{0}$ is the unique equilibrium and the system $(\mathbf{x}(t),\mathbf{y}(t))$ is exponentially stable at the origin for every initial state satisfying the model conditions.
\end{proof}

\subsection{Proof of Theorem~\ref{sdir_previous_relation}}
\begin{proof}
From the choice of $q_i$ $\forall i \in [N]$ that we have fixed, $D_i\leq D'_i+(1-q_i^{-1})W_i$ for all $i$, which yields $
\mathbf{W}+\mathbf{Q}(\mathbf{I}-\mathbf{W}-\mathbf{D}')
\leq
(\mathbf{I}-\mathbf{D})\mathbf{Q}$. Moreover, since $\mathbf{G}_Q\mathbf{F}\geq\mathbf{0}$, we have
$(\mathbf{I}-\mathbf{D})\mathbf{Q}
\leq
\mathbf{M}_{Q}\mathbf{Q}$.
Hence,
$\mathbf{L}_Q\mathbf{M}
\leq
\mathbf{M}_{Q}\mathbf{L}_Q$,
where
$\mathbf{L}_Q=[\mathbf{I}\ \mathbf{Q}]$.

Since $\mathbf{M}\geq0$, the Perron--Frobenius theorem guarantees the existence of
$\mathbf{z}\geq0$, $\mathbf{z}\neq\mathbf{0}$ such that
$\mathbf{M}\mathbf{z}=\rho(\mathbf{M})\mathbf{z}$.
Let $\mathbf{u}=\mathbf{L}_Q\mathbf{z}$.
Since $q_i>0$, we have $\mathbf{u}\geq0$ and
$\mathbf{u}\neq\mathbf{0}$. From the above inequality, we obtain $\mathbf{M}_Q\mathbf{u}
\geq
\mathbf{L}_Q\mathbf{M}\mathbf{z}
=
\rho(\mathbf{M})\mathbf{u}$. By the subinvariance property of the Perron--Frobenius theorem
\cite[Chapter~8]{Horn2012}, it follows that
$
\rho(\mathbf{M})\leq\rho(\mathbf{M}_Q).
$
\end{proof}

\subsection{Proof of Theorem~\ref{upper-bound}}

\begin{proof}
As in the proof of Theorem~\ref{sdir_convergence}, since
$\mathbf{S}(t)\leq\mathbf{S}(0)$, we have
$
\mathbf{v}(t+1)
\leq
\mathbf{M}_{-P}\mathbf{v}(t).
$
Since $\mathbf{M}_{-P}\geq0$,
$\mathbf{v}(t)\leq\mathbf{M}_{-P}^t\mathbf{v}(0)$.
Since $\rho(\mathbf{M}_{-P})<1$,
$\sum_{t=0}^{\infty}\mathbf{M}_{-P}^t
=(\mathbf{I}_{2N}-\mathbf{M}_{-P})^{-1}$, and hence
$\sum_{t=0}^{\infty}\mathbf{v}(t)
\leq
(\mathbf{I}_{2N}-\mathbf{M}_{-P})^{-1}\mathbf{v}(0)$. Let
$
\begin{bmatrix}
\mathbf{u}_x\\
\mathbf{u}_y
\end{bmatrix}
=
(\mathbf{I}_{2N}-\mathbf{M}_{-P})^{-1}
\begin{bmatrix}
\mathbf{x}(0)\\
\mathbf{y}(0)
\end{bmatrix}.
$ Then,
$\sum_{t=0}^{\infty}\mathbf{x}(t)\leq\mathbf{u}_x$.

We have
$
\mathbf{I}_{2N}-\mathbf{M}_{-P}
=
\begin{bmatrix}
\mathbf{D}-\mathbf{A}\mathbf{F}_{-P}
&
-\mathbf{W}\\
-(\mathbf{I}-\mathbf{A})\mathbf{F}_{-P}
&
\mathbf{W}+\mathbf{D}'
\end{bmatrix}.
$ Therefore,
$(\mathbf{D}-\mathbf{A}\mathbf{F}_{-P})\mathbf{u}_x
-\mathbf{W}\mathbf{u}_y=\mathbf{x}(0)$, and
$-(\mathbf{I}-\mathbf{A})\mathbf{F}_{-P}\mathbf{u}_x
+(\mathbf{W}+\mathbf{D}')\mathbf{u}_y=\mathbf{y}(0)$. Under Assumption~\ref{assp3}, $\mathbf{W}+\mathbf{D}'$ is invertible.
From the second equation,
$
\mathbf{u}_y
=
(\mathbf{W}+\mathbf{D}')^{-1}
\left(
\mathbf{y}(0)
+
(\mathbf{I}-\mathbf{A})\mathbf{F}_{-P}\mathbf{u}_x
\right).
$
Substituting into the first equation and using
$\mathbf{T}=\mathbf{W}(\mathbf{W}+\mathbf{D}')^{-1}$ yields
$
(\mathbf{D}-\mathbf{\Theta}\mathbf{F}_{-P})\mathbf{u}_x
=
\mathbf{x}(0)+\mathbf{T}\mathbf{y}(0).
$ We have $
\det(\mathbf{I}_{2N}-\mathbf{M}_{-P})
=
\det(\mathbf{W}+\mathbf{D}')
\det(\mathbf{D}-\mathbf{\Theta}\mathbf{F}_{-P})$.
Since $\rho(\mathbf{M}_{-P})<1$, $\det(\mathbf{I}_{2N}-\mathbf{M}_{-P})\neq 0$. Together with the invertibility of $\mathbf{W}+\mathbf{D}'$, this implies that $\mathbf{D}-\mathbf{\Theta}\mathbf{F}_{-P}$ is invertible. Therefore,
$
\mathbf{u}_x
=
(\mathbf{D}-\mathbf{\Theta}\mathbf{F}_{-P})^{-1}
\left(
\mathbf{x}(0)+\mathbf{T}\mathbf{y}(0)
\right).
$
Thus,
\[
\sum_{t=0}^{\infty}\mathbf{x}(t)
\leq
(\mathbf{D}-\mathbf{\Theta}\mathbf{F}_{-P})^{-1}
\left(
\mathbf{x}(0)+\mathbf{T}\mathbf{y}(0)
\right).
\]

Moreover,
\begin{align*}
\mathbf{m}(t)-\mathbf{m}(0)
&=
\sum_{\tau=0}^{t-1}
\mathbf{S}(\tau)\mathbf{B}_{-P}\mathbf{x}(\tau)\leq
\mathbf{F}_{-P}
\sum_{\tau=0}^{t-1}\mathbf{x}(\tau).
\end{align*}
Letting $t\rightarrow\infty$ gives
$
\mathbf{m}^*-\mathbf{m}(0)
\leq
\mathbf{F}_{-P}\sum_{t=0}^{\infty}\mathbf{x}(t).
$
Hence, from \eqref{cumulative},
\begin{align*}
\Phi(P)
&=
\|\mathbf{m}^*-\mathbf{m}(0)\|_1\\
&\leq
1^\top\mathbf{F}_{-P}
(\mathbf{D}-\mathbf{\Theta}\mathbf{F}_{-P})^{-1}
\left(
\mathbf{x}(0)+\mathbf{T}\mathbf{y}(0)
\right)=\widehat{\Phi}(P).
\end{align*}
\end{proof}
\subsection{Proof of Theorem~\ref{tighter-bound}}

\begin{proof}
Similar to the proof of Theorem~\ref{sdir_previous_relation}, we can prove that $
\rho(\mathbf{M}_{Q,-P})
\geq
\rho(\mathbf{M}_{-P})
$.
Since $\rho(\mathbf{M}_{Q,-P})<1$, we have $\rho(\mathbf{M}_{-P})<1$, and Theorem~\ref{upper-bound} gives $
\Phi(P)\leq\widehat{\Phi}(P).
$

Moreover, $T_i=\frac{W_i}{W_i+D'_i}\leq q_i$ for all $i\in[N]$, so
$\mathbf{T}\leq\mathbf{Q}$ and
$\mathbf{\Theta}
=\mathbf{A}+\mathbf{T}(\mathbf{I}-\mathbf{A})
\leq
\mathbf{A}+\mathbf{Q}(\mathbf{I}-\mathbf{A})
=\mathbf{G}_Q$.
Let
$\mathbf{Z}_T=\mathbf{D}-\mathbf{\Theta}\mathbf{F}_{-P}$ and
$\mathbf{Z}_Q=\mathbf{D}-\mathbf{G}_Q\mathbf{F}_{-P}$.
Then $\mathbf{Z}_T\geq\mathbf{Z}_Q$.
Since $\rho(\mathbf{M}_{Q,-P})<1$ and
$\rho(\mathbf{M}_{-P})<1$, we have
$\mathbf{Z}_Q^{-1}\geq0$ and $\mathbf{Z}_T^{-1}\geq0$.
From
$\mathbf{Z}_Q^{-1}-\mathbf{Z}_T^{-1}
=
\mathbf{Z}_Q^{-1}
(\mathbf{Z}_T-\mathbf{Z}_Q)
\mathbf{Z}_T^{-1}\geq0$,
it follows that $\mathbf{Z}_T^{-1}\leq\mathbf{Z}_Q^{-1}$.

Moreover, since $\mathbf{T}\leq\mathbf{Q}$, $\mathbf{x}(0)+\mathbf{T}\mathbf{y}(0)
\leq
\mathbf{x}(0)+\mathbf{Q}\mathbf{y}(0)$.
Therefore,
\begin{align*}
\widehat{\Phi}(P)
&=
1^\top\mathbf{F}_{-P}\mathbf{Z}_T^{-1}
\left(\mathbf{x}(0)+\mathbf{T}\mathbf{y}(0)\right)\\
&\leq
1^\top\mathbf{F}_{-P}\mathbf{Z}_Q^{-1}
\left(\mathbf{x}(0)+\mathbf{Q}\mathbf{y}(0)\right)=\Phi_Q(P).
\end{align*}
Thus
$\Phi(P)\leq\widehat{\Phi}(P)\leq\Phi_Q(P)$.
\end{proof}

\subsection{Proof of Proposition~\ref{supermodular}}
\begin{proof}
Let
$\mathbf{\Xi}_{-P}
=\mathbf{D}-\mathbf{\Theta}\mathbf{F}_{-P}$,
$H(P)=\mathbf{\Xi}_{-P}^{-1}$, and
$\mathbf{z}=\mathbf{x}(0)+\mathbf{T}\mathbf{y}(0)$. Consider an edge $a=(j,i)\notin P$ and let
$c=S_{ii}(0)B_{ij}\geq0$ and $d=c\Theta_{ii}\geq0$.
After adding $a$ to the deleted edge set, we have
$\mathbf{F}_{-(P\cup\{a\})}
=\mathbf{F}_{-P}-c\mathbf{u}_i\mathbf{u}_j^\top$ and
$\mathbf{\Xi}_{-(P\cup\{a\})}
=\mathbf{\Xi}_{-P}+d\mathbf{u}_i\mathbf{u}_j^\top$.
By the Sherman--Morrison formula,
\[
H(P\cup\{a\})
=
H(P)
-
\frac{
dH(P)\mathbf{u}_i\mathbf{u}_j^\top H(P)
}{
1+d\mathbf{u}_j^\top H(P)\mathbf{u}_i
}.
\]

Since $\mathbf{M}_{-P}\leq\mathbf{M}$ and $\rho(\mathbf{M})<1$, we have
$\rho(\mathbf{M}_{-P})<1$. Hence,
$(\mathbf{I}_{2N}-\mathbf{M}_{-P})^{-1}
=\sum_{k=0}^{\infty}\mathbf{M}_{-P}^{k}\geq\mathbf{0}$.
By the block inverse formula,
$H(P)=\mathbf{\Xi}_{-P}^{-1}$ is the upper-left block of
$(\mathbf{I}_{2N}-\mathbf{M}_{-P})^{-1}$, and therefore
$H(P)\geq\mathbf{0}$.
Thus, $H(P\cup\{a\})\leq H(P)$, i.e., $H(P)$ is entrywise
monotone non-increasing. Next,
\begin{align*}
H(P)-H(P\cup\{a\})
&=
\int_0^1
d\,
(\mathbf{\Xi}_{-P}+\lambda d\mathbf{u}_i\mathbf{u}_j^\top)^{-1}
\mathbf{u}_i\mathbf{u}_j^\top\\
&\quad\cdot
(\mathbf{\Xi}_{-P}+\lambda d\mathbf{u}_i\mathbf{u}_j^\top)^{-1}
\,d\lambda.
\end{align*}
Consider $P_1\subseteq P_2$ and $a\notin P_2$.
For every $\lambda\in[0,1]$,
$\mathbf{F}_{-P_1}-\lambda c\mathbf{u}_i\mathbf{u}_j^\top
\geq
\mathbf{F}_{-P_2}-\lambda c\mathbf{u}_i\mathbf{u}_j^\top$.
The corresponding comparison matrices are nonnegative and bounded
entrywise by $\mathbf{M}$. Hence, their spectral radii are less than one, and by the Neumann series \cite[Chapter~8]{Horn2012}, $ (\mathbf{\Xi}_{-P_1} +\lambda d\mathbf{u}_i\mathbf{u}_j^\top)^{-1} \geq (\mathbf{\Xi}_{-P_2} +\lambda d\mathbf{u}_i\mathbf{u}_j^\top)^{-1}.$
Since these matrices are nonnegative, it follows that
$
H(P_1)-H(P_1\cup\{a\})
\geq
H(P_2)-H(P_2\cup\{a\}).
$
Therefore, $H(P)$ is entrywise supermodular.

Finally, using
$\mathbf{F}_{-(P\cup\{a\})}
=\mathbf{F}_{-P}-c\mathbf{u}_i\mathbf{u}_j^\top$, we obtain
$
\widehat{\Phi}(P)-\widehat{\Phi}(P\cup\{a\})
=
\mathbf{1}^\top\mathbf{F}_{-P}
\bigl(H(P)-H(P\cup\{a\})\bigr)\mathbf{z}
+
c\mathbf{1}^\top\mathbf{u}_i\mathbf{u}_j^\top
H(P\cup\{a\})\mathbf{z}.
$
Since all matrices and vectors above are nonnegative,
$\widehat{\Phi}(P)\geq\widehat{\Phi}(P\cup\{a\})$, so
$\widehat{\Phi}(P)$ is monotone non-increasing.

Moreover, for $P_1\subseteq P_2$ and $a\notin P_2$, we have
$\mathbf{F}_{-P_1}\geq\mathbf{F}_{-P_2}$,
$H(P_1)-H(P_1\cup\{a\})
\geq
H(P_2)-H(P_2\cup\{a\})$, and
$H(P_1\cup\{a\})\geq H(P_2\cup\{a\})$.
Therefore,
$
\widehat{\Phi}(P_1)-\widehat{\Phi}(P_1\cup\{a\})
\geq
\widehat{\Phi}(P_2)-\widehat{\Phi}(P_2\cup\{a\}).
$
Hence, $\widehat{\Phi}(P)$ is monotone non-increasing and
supermodular with respect to $P\subseteq\mathcal{Q}$.
\end{proof}

\begin{thebibliography}{99}%
	\bibitem{Albert2000}
R. Albert, H. Jeong, and A.-L. Barabási,
``Error and attack tolerance of complex networks,''
Nature, vol. 406, no. 6794, pp. 378--382, 2000.

\bibitem{Ahn2013}
H. J. Ahn and B. Hassibi,
``Global dynamics of epidemic spread over complex networks,''
in Proc. 52nd IEEE Conf. Decision and Control (CDC),
2013, pp. 4579--4585.

\bibitem{Callaway2000}
D. S. Callaway, M. E. J. Newman, S. H. Strogatz, and D. J. Watts,
``Network robustness and fragility: Percolation on random graphs,''
Phys. Rev. Lett., vol. 85, pp. 5468--5471, 2000.

\bibitem{Chen2021}
L. H. Chen, L. Hung, H. Lotze, and P. Rossmanith,
``Online node- and edge-deletion problems with advice,''
Algorithmica, vol. 83, pp. 2719--2753, 2021.

\bibitem{Chen2014}
W. Chen, C. Castillo, and L. V. Lakshmanan,
\textit{Information and Influence Propagation in Social Networks}.
Morgan \& Claypool Publishers, 2014.

\bibitem{Cho2024}
D. X. Cho, T. H. Anh, N. T. L. Phuong, and N. K. Khoa,
``Two-stage APT malware propagation model in computer networks,''
Neural Comput. Appl., vol. 37, pp. 21805--21832, 2025.

\bibitem{Dung2026}
H. P. Dung and D. K. Ly,
``Minimizing cumulative infections in SIS epidemic models over networks via an edge deletion algorithm,''
in Proc. 11th Int. Conf. Micro-Electronics, Electromagnetics and Telecommunications (ICMEET),
Lecture Notes in Electrical Engineering, Springer, to appear, 2026.

\bibitem{Dung2026-2}
H. P. Dung and N. H. Phuc,
``A novel approach for epidemic threshold of networks,''
2026, arXiv:2607.14048.

\bibitem{Firth2020}
J. A. Firth, J. Hellewell, P. Klepac, S. Kissler, A. J. Kucharski,
and L. G. Spurgin,
``Using a real-world network to model localized COVID-19 control strategies,''
Nat. Med., vol. 26, pp. 1616--1622, 2020.

\bibitem{Horn2012}
R. A. Horn and C. R. Johnson,
\textit{Matrix Analysis}.
Cambridge University Press, 2012.

\bibitem{Kempe2003}
D. Kempe, J. M. Kleinberg, and E. Tardos,
``Maximizing the spread of influence through a social network,''
in Proc. 9th ACM SIGKDD Int. Conf. Knowledge Discovery and Data Mining (KDD),
2003, pp. 137--146.

\bibitem{Kempe2015}
D. Kempe, J. Kleinberg, and E. Tardos,
``Maximizing the spread of influence through a social network,''
Theory Comput., vol. 11, no. 4, pp. 105--147, 2015.

\bibitem{Kermack1927}
W. O. Kermack and A. G. McKendrick,
``A contribution to the mathematical theory of epidemics,''
Proc. Roy. Soc. London A, vol. 115, no. 772, pp. 700--721, 1927.

\bibitem{Khanh2026}
T. V. Khanh, D. X. Cho, and H. P. Dung,
``A novel discrete-time model of information diffusion on social networks considering users behavior,''
in Proc. 40th Int. Conf. Infor. Networking (ICOIN),
IEEE, 2026, pp. 650--655.

\bibitem{Klepac2018}
P. Klepac, S. Kissler, and J. Gog,
``Contagion! The BBC Four Pandemic--the model behind the documentary,''
Epidemics, vol. 24, pp. 49--59, 2018.

\bibitem{Liu2022}
J. Liu, T. Saeed, and A. Zeb,
``Delay effect of an e-epidemic SEIRS malware propagation model with a generalized non-monotone incidence rate,''
Results Phys., vol. 39, Art. no. 105672, 2022.

\bibitem{Mieghem2009}
P. Van Mieghem,
``Virus spread in networks,''
IEEE/ACM Trans. Netw., vol. 17, no. 1, pp. 1--14, 2009.

\bibitem{Nowzari2016}
C. Nowzari, V. M. Preciado, and G. J. Pappas,
``Analysis and control of epidemics: A survey of spreading processes on complex networks,''
IEEE Control Syst., vol. 36, pp. 26--46, 2016.

\bibitem{Pare2020}
P. E. Pare, J. Liu, C. Beck, B. Kirwan, and T. Basar,
``Analysis, estimation, and validation of discrete-time epidemic processes,''
IEEE Trans. Control Syst. Technol., vol. 28, no. 1, pp. 79--93, 2020.

\bibitem{Pastor-Satorras2001}
R. Pastor-Satorras and A. Vespignani,
``Epidemic spreading in scale-free networks,''
Phys. Rev. Lett., vol. 86, pp. 3200--3203, 2001.

\bibitem{Pastor-Satorras2015}
R. Pastor-Satorras, C. Castellano, P. Van Mieghem, and A. Vespignani,
``Epidemic processes in complex networks,''
Rev. Mod. Phys., vol. 87, pp. 925--979, 2015.

\bibitem{Pham2019}
C. V. Pham, Q. V. Phu, H. X. Hoang, J. Pey, and M. T. Thai,
``Minimum budget for misinformation blocking in online social networks,''
J. Comb. Optim., vol. 38, pp. 1101--1127, 2019.

\bibitem{Ruhi2015}
A. Ruhi and B. Hassibi,
``SIRS epidemics on complex networks: Concurrence of exact Markov chain and approximated models,''
in Proc. 54th IEEE Conf. Decision and Control (CDC),
2015, pp. 2919--2926.

\bibitem{Shakarian2015}
P. Shakarian, A. Bhatnagar, A. Aleali, E. Shaabani, and R. Guo,
\textit{Diffusion in Social Networks}.
Springer, 2015.

\bibitem{Sharkey2008}
K. Sharkey,
``Deterministic epidemiological models at the individual level,''
J. Math. Biol., vol. 57, pp. 311--331, 2008.

\bibitem{Silva2016}
T. C. Silva and L. Zhao,
\textit{Machine Learning in Complex Networks}.
Cham, Switzerland: Springer, 2016.

\bibitem{Xie2023}
J. Xie, F. Zhang, K. Wang, X. Lin, and W. Zhang,
``Minimizing the influence of misinformation via vertex blocking,''
in Proc. IEEE Int. Conf. Data Engineering (ICDE),
2023, pp. 789--801.

\bibitem{Wang2003}
Y. Wang, D. Chakrabarti, C. Wang, and C. Faloutsos,
``Epidemic spreading in real networks: An eigenvalue viewpoint,''
in Proc. 22nd Int. Symp. Reliable Distributed Systems (SRDS),
2003, pp. 25--34.

\bibitem{Yi2022}
Y. Yi, L. Shan, P. Pare, and K. H. Johansson,
``Edge deletion algorithms for minimizing spread in SIR epidemic models,''
SIAM J. Control Optim., vol. 60, no. 2, pp. 246--273, 2022.

\bibitem{Youssef2011}
M. Youssef and C. Scoglio,
``An individual-based approach to SIR epidemics in contact networks,''
J. Theor. Biol., vol. 283, pp. 136--144, 2011.
\end{thebibliography}
\end{document}